\documentclass[12pt,draftclsnofoot,onecolumn]{IEEEtran}
\usepackage{cite}
\usepackage{url}

\usepackage[dvipsnames,svgnames]{xcolor}
\usepackage{graphicx}
\usepackage{algorithmic}
\usepackage{psfrag}
\DeclareGraphicsExtensions{.xps}
\usepackage{epstopdf} 

\ifCLASSOPTIONcompsoc
\usepackage[caption=false,font=normalsize,labelfon
t=sf,textfont=sf]{subfig}
\else
\usepackage[caption=false,font=footnotesize]{subfig}
\fi

\usepackage[cmex10]{amsmath}
\usepackage{amssymb}
\usepackage{amsthm}

\usepackage{tablefootnote}

\newtheorem{theorem}{Theorem}
\newtheorem{lemma}{Lemma}
\newtheorem{corollary}{Corollary}

\usepackage{eso-pic}

\newcommand{\vect}[1]{\mathbf{#1}}

\def\tr{\mathrm{tr}}

\def\Htran{\mbox{\tiny $\mathrm{H}$}}

\def\CN{\mathcal{N}_{\mathbb{C}}} 

\begin{document}

\title{Spectral Efficiency of Dense Multicell Massive MIMO Networks in Spatially Correlated Channels}
\author{
\IEEEauthorblockN{FahimeSadat Mirhosseini\IEEEauthorrefmark{1}, Aliakbar Tadaion\IEEEauthorrefmark{1}, S. Mohammad Razavizadeh\IEEEauthorrefmark{2}}\\
\IEEEauthorblockA{\IEEEauthorrefmark{1}\small{Department of Electrical Engineering, Yazd University, Yazd, Iran}}\\
\IEEEauthorblockA{\IEEEauthorrefmark{2}\small{School of Electrical Engineering, Iran University of Science \& Technology, Tehran, Iran}}
}

\maketitle

\begin{abstract}
This paper is on the spectral efficiency (SE) of a dense multi-cell massive multiple-input multiple-output (MIMO).  The channels are spatially correlated and the multi-slope path loss model is considered. 
In our framework, the channel state information is obtained by using pilot sequences and the BSs are deployed randomly. 
First, we study the channel estimation accuracy and its impact on the SE as the BS density increases and the network becomes densified. Second, we consider the special case of uncorrelated channels for which the stochastic geometry framework helps us to simplify the SE expressions, and obtain the minimum value of antenna-UE ratio over which the pilot contamination is dominant rather than the inter- and intra-cell interference.
Finally, we provide some insights into the obtained SE for the spatially correlated channels, from a multi-cell processing scheme as well as the single-cell ones in terms of the BS density. Our results show that while all the detectors result in non-increasing SE in terms of the BS density, their area SE increases exponentially as the network becomes densified. 
Moreover, we conclude that in order to achieve a given SE, the required value of antenna-UE ratio decreases as the level of channel correlation increases.
\end{abstract}


\section{Introduction}
The next generation of the wireless cellular communications is expected to encounter an explosive demand for mobile data traffic. One way to meet these demands is \textit{network densification}, which means to reduce the cell sizes by deploying more base station (BS)s in a certain geographical area (i.e., higher BS density) \cite{AndrewsMagazine}. 
Dense networks can be modeled by a homogeneous Poisson point process (H-PPP) and analyzed by stochastic geometry tools \cite{Andrews2014a}.
Another promising technology towards the capacity improvement  is \textit{massive multiple-input multiple-output (MIMO)} \cite{marzetta2010noncooperative,Larsson2014,xin2015area}, wherein the BSs are equipped with a large number $M$ of 
antennas to serve a multitude $K \gg 1$ user equipment (UE)s by spatial multiplexing while $K \ll M$. These two schemes can provide significant area spectral efficiency (ASE), {defined as the sum of the 
spectral efficiency (SE)\footnote{The SE is defined as the maximum data rate (capacity) per unit bandwidth
in a specific communication system.
} of all UEs per unit area \cite{alouini1999area, mirhosseini2015performance}.}

The vast majority of massive MIMO literature considers the network with two major simplifying assumptions which make the network analysis tractable \cite{pizzo2018network, bjornson2016deploying, rusek2013scaling, fahim}.
First, the propagation channels to the multi-antenna BSs are assumed spatially uncorrelated; whereas the practical channels are generally spatially correlated \cite{Gao2015a, Gao2015b, le2019review}. 
Second simplifying assumption is that the single-cell processing schemes designed for the single-cell operation
can be easily employed for the multi-cell network.  Whereas the multi-cell processing schemes like multi-cell minimum mean squared error (M-MMSE) can achieve better performance in all scenarios. The M-MMSE is the optimum multi-cell processing scheme which maximizes the SINR and substantially improves the ASE \cite{massivemimobook}.
However, there are some literature which avoid these simplifying assumptions to show the network performance \cite{hoydis2013massive, adhikary2017uplink, You2015a, EmilEURASIP17, SanguinettiBH19}. 
{For example, in \cite{hoydis2013massive}, the authors consider the UL and DL of the multi-cell network with the correlated channels, and then investigate the number of antennas required to achieve the fraction of the ultimate performance limit by considering maximum ratio (MR) and single-cell MMSE (S-MMSE).}
Utilizing the correlated feature of fading, a pilot reuse scheme is proposed in \cite{You2015a} to reduce the pilot overhead. On the other hand, the pilot contamination impact on the massive MIMO can be resolved by employing M-MMSE in  the spatially correlated channel \cite{BHS18A,SanguinettiBH19}.

Moreover, the densified network analysis is sensitive to the path loss model.
In the dense network, a more realistic model accounts for a multi-slope path loss, wherein the path loss exponent varies with the distance between UEs and BSs \cite{zhang2015downlink,pizzo2018network}.
{If  a distance-independent or single-slope path loss model is adopted, the ASE is a linearly increasing function of the BS density for single-input single-output 
networks\cite{andrews2011tractable}. Whereas,  we may observe different trends in the ASE ranging from linear growth to collapse under network densification by adopting more realistic path loss model \cite{alammouri2018unified, ding2017performance}.
Moreover, an analytical expression for the energy efficiency  is obtained in a 
 multi-slope model using some approximations in \cite{pizzo2018network} and the coverage probability is investigated  in \cite{Kountouris2016} for dual-slope.}

{However, it should be noted  that the impact of network densification (i.e., increasing BS density) is not addressed  for the mentioned general and also practical framework in the recent works. Motivated by the above discussion, this paper investigates the UL of a dense massive MIMO network in which the $M$-antenna BSs are distributed in a given area according to an H-PPP of a determined density and each of them serves $K$ single-antenna UEs. The channel is spatially correlated and the path loss has a multi-slope model. The channel state information (CSI) is estimated using some orthogonal pilots which are reused in different cells with a pilot reuse factor. In this network, the SE and ASE performances are evaluated in terms of BS density under single-cell and multi-cell detection schemes. Since the SE per UE
 is affected by the channel estimation error, we first investigate the behavior of the channel estimation accuracy based on the BS density with different values of pilot reuse factor. In the special case when the environment has non-line-of-sight (NLOS) propagation and the BS experiences many scattering objects, the channel model is simplified  to spatially uncorrelated fading. In this case, the lower bound on the SE obtained for MR and zero-forcing (ZF), helps us to analytically provide the interplay between the network parameters, such as antenna-UE ratio or $M/K$, BS density and pilot reuse factor. Remarkably, based on the network parameters, we obtain the region for which the pilot contamination dominates the inter- and intra-cell interference. We also realize that the SE of ZF has larger reduction rate than the SE of MR with respect to the BS density.
Moreover, we observe that in the dense network\footnote{In order to characterize the network performance based on the BS density, we consider different degrees of BS densification \cite{ge20165g}: low dense with roughly a BS density $\le 10\ {\rm BS/km^2}$ and dense with the BS density $ > 10\ {\rm BS/km^2}$.}, SE obtained by MR and ZF does not converge to their ultimately achievable rates (i.e., with infinitely many antennas, $M \rightarrow \infty$) even through increasing $M / K$ (for example up to 50). This issue can be addressed by employing multi-cell processing, M-MMSE which is designed to reduce the interference from the UEs in all cells. In the correlated channel, we show that M-MMSE has the superior performance compared with the single-cell detectors specially in the dense networks wherein a large level of interference is experienced. Moreover, the interplay between the network parameters, such
as $M / K$ and BS density with the level of channel correlation is derived. Our main contributions can be summarized as follows:
\begin{itemize}
\item
We derive the analytic form of the channel estimation accuracy and evaluate its behavior in terms of BS density and pilot reuse factor.
\item
We analyze the SE and ASE performance of a multi-cell massive MIMO network in a more realistic channel model which includes spatial correlation in each user's channel and is adopted with a multi-slope path loss model.
\item
We study both low dense and dense massive MIMO networks and investigate the effect of network densification on the network performance (in terms of SE and ASE). 
\item
We compare the behavior of single-cell and multi-cell processing schemes in terms of network densification.
\end{itemize}}

The remaining of the paper is organized as follows. In Section \ref{sec:Problem_Statement}, we define the network model with pilot allocation scheme, power control policy and channel estimation. In Section \ref{sec:net_analyz}, we compute the channel estimation accuracy in both correlated and uncorrelated Rayleigh fading in terms of the BS density and provide the SE and ASE for both types of channels with single-cell and multi-cell combining schemes. 
We illustrate the simulation results in Section \ref{sec:simulationresults} and finally, we conclude the paper in Section \ref{sec:conclusion}.

{\textit{Notation}: We use upper and lower bold face letters for matrices and column vectors, respectively. The operator $\left(. \right)^{H}$ refers to conjugate transpose and Euclidean norm vector operator is shown by  $||.||$.
The $M \times M$ identity matrix is denoted by $\mathbf{I}_M$ and $\mathbf{0}$ is the zero vector. {We use $\mathbb{E}_{\mathbf{x}}\{ .\}$ as the expectation operator with respect to the
random vector $\mathbf{x}$.}
Moreover, we have $\Gamma(a, x) = \int_{x}^{\infty} e^{-t} t^{a - 1} dt$ as the upper incomplete gamma function. $\mathbb{R}^2$ and $\mathbb{N}$ denote the two-dimensional real-valued and positive integer-valued vector spaces.}


\section{Network Model}\label{sec:Problem_Statement}

We consider the uplink of a massive MIMO network wherein each BS has $M \gg 1$ antennas and serves $K$ single-antenna UEs by the nearest BS association rule. The BSs are spatially distributed at locations ${\vect{x}_l\subset \mathbb{R}^2}$ according to an H-PPP ${\vect{\Psi}_{\lambda} = \{\vect{x}_l; \, l \in \mathbb{N} \}}$ of density ${\lambda}$ $[\rm{BS/km^2}]$ while their $K$ connected UEs are uniformly distributed in their coverage area which is a Poisson-Voronoi cell as shown in Fig.~\ref{fig:SGmodel}.
The network operates according to the synchronous time-division-duplex protocol over each coherence block composed of $\tau_c = B_c T_c$ samples where $B_c$ and $T_c$ are the channel coherence bandwidth and channel coherence time, respectively. The CSI is acquired using $\tau_p = K \zeta$ UL pilot sequences in each coherence block with $\zeta \geq 1$ is the pilot reuse factor.

\begin{figure}
 \centering
 \includegraphics[width=0.6\textwidth]{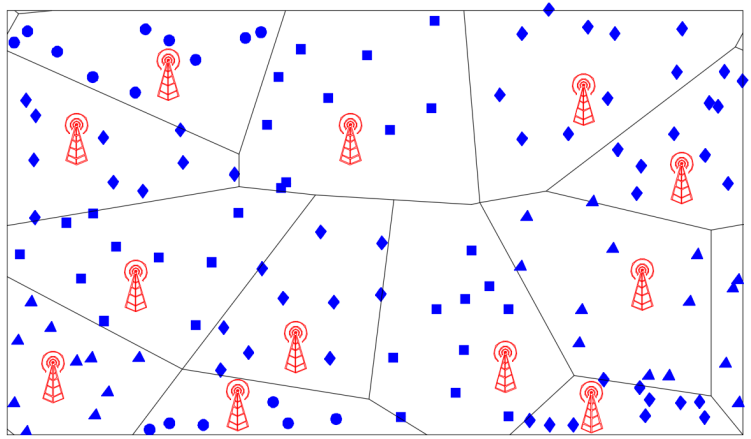} 
 \caption{{Deployment of a cellular network with BSs drawn from a H-PPP, $\vect{\Psi}_{\lambda}$ with $K = 10$ and $\zeta = 4$.}}
 \label{fig:SGmodel}
 \end{figure}

\subsection{Channel model}\label{subsec:channel_model}
The channel $\vect{h}_{li}^{j} \in \mathbb{C}^{M}$ between the UE~$i$ in the cell~$l$ and the BS~$j$ is modeled as correlated Rayleigh fading:
\begin{align}\label{eq:iid-Rayleigh fading}
\vect{h}_{li}^{j} \sim \CN( \vect{0}_M, \vect{R}_{li}^{j}),
\end{align}
{where the Gaussian distribution accounts for the random
small-scale fading realization in each coherence block, and
 $\vect{R}_{li}^{j} \in \mathbb{C}^{M \times M}$ is the spatial correlation matrix. {The correlation matrix is known to
change on a time-scale much larger than the coherence time
of the fading channel; hence, for the sake of simplicity,  we assume that  $\vect{R}_{li}^{j}$ is fixed and known at the BS}. Practical methods for its estimation can be found in \cite{Bjornson2016c,Vorobyov2018}.
{In general, there are various correlation models proposed in the literature; for example, the equal, exponential, Bessel and one-ring correlation models  \cite{le2019review, BHS18A, SanguinettiBH19}.
In this paper, we use the most common 2D one-ring channel model for a uniform linear array (ULA)  with half-wavelength spacing and average path loss $\beta_{li}^j$.} For an angle-of-arrival (AoA) $\varphi_{li}^j$, the scatterers are uniformly distributed in $[\varphi_{li}^j -\frac{\Delta}{2}, \varphi_{li}^j + \frac{\Delta}{2}]$ with $\Delta$ being the angular spread. The 
element of $\vect{R}_{li}^j$ is \cite{massivemimobook}:
\begin{align}\label{eq:2DChannelModel}
 \left[ \vect{R}_{li}^j(d_{li}^j, \varphi_{li}^j) \right]_{m_1,m_2}  = \frac{\beta_{li}^j(d_{li}^{j})}{2\Delta} \int_{-\Delta}^{\Delta}{  e^{\mathsf{j} \pi(m_1-m_2) \sin(\varphi_{li}^j + {\varphi}) }}d{\varphi}.
\end{align}
{where  $d_{li}^j$ is the average distance of the UE $i$ in the cell
$l$ from the BS $j$.} 
\begin{corollary}
{In the special case that the scatterers are uniformly distributed in the interval $[-\pi, \pi]$, the channel is modeled as an uncorrelated Rayleigh fading with the channel correlation matrix  $\mathbf{R}_{li}^{j} = \beta_{li}^{j}(d_{li}^{j}) \mathbf{I}_M$}.
\end{corollary}


The normalized trace $\beta_{li}^{j}(d_{li}^{j}) = \frac{1}{M}\tr(\vect{R}_{li}^{j}(d_{li}^j, \varphi_{li}^j))$ represents the average channel gain from a generic antenna at the BS $j$ to the UE $i$ in the cell $l$ \cite{SanguinettiBH19}. 
We consider the realistic path loss model for the dense network for which  $\beta_{li}^{j}(d_{li}^{j})$ is computed according to a multi-slope path loss model  \cite{zhang2015downlink} where
$\beta_{li}^{j}(d_{li}^{j})  = \beta_{n,li}^{j}(d_{li}^{j}) =  \Upsilon_n (d_{li}^{j})^{-\alpha_n}$ with $r_{n-1} \leq d_{li}^{j} \leq r_n$
denotes the distance of the UE $i$ in the cell $l$ from the BS $j$, ${\alpha_1 \leq \cdots \leq \alpha_N}$ are the path loss exponents, ${0 = r_0 < \cdots < r_N = \infty}$ refer to the distances at which a change in the power decadence occurs. Finally, the coefficients ${\Upsilon_{n+1} = \Upsilon_n r_{n}^{\alpha_{n+1} - \alpha_{n}}}$ for $n = 1, \ldots, N-1$ are selected to satisfy the continuity purposes of the model with $\Upsilon_1$ being a design parameter. The widely used single-slope path loss model can be obtained by $N = 1$,  i.e., ${\beta_{li}^{j} = \Upsilon_1 (d_{li}^{j})^{-\alpha_1}}$. Note that for the sake of simplicity, we remove all the dependency parameters and just consider $\beta_{li}^j$ instead of $\beta_{li}^{j}(d_{li}^{j})$ and  $\vect{R}_{li}^{j}$ instead of $\vect{R}_{li}^{j}(d_{li}^j, \varphi_{li}^j)$ .
{We use the the dual-slope path loss model, i.e.,  $N=2$, given in \cite{fahim} with parameters reported in Table~\ref{table1} in this paper. 
As discussed in \cite{AndrewsMagazine} and \cite{zhang2015downlink}, this model is a simplified version of the ITU-R UMi model \cite{3GPP}.}


\subsection{Pilot reuse policy}\label{subsec:pilotreuse}
The channel estimation in each cell $l$ is performed based on a pilot book $\{{\boldsymbol{\Phi}} \in \mathbb{C}^{\tau_p \times \tau_p}\}$ of $\tau_p$ orthonormal uplink pilot sequences ${\{\boldsymbol{\phi}_{li}\in \mathbb{C}^{\tau_p}\}_{i=1}^{\tau_p}}$ with $\boldsymbol{\phi}_{li}^H\boldsymbol{\phi}_{li}=\tau_p$.
We assume that in each coherence block, the BS $l$ picks uniformly a subset of $K$ different sequences from $\boldsymbol{\Phi}$ and assigns them to its served UEs 
\cite{bjornson2016deploying}. This implies that a pilot reuse factor of ${\zeta \triangleq \tau_p / K \geq 1}$ is used throughout the network. 
This pilot allocation could be modeled in each cell by a Bernoulli random variable ${a_{l' l} \sim \mathcal{B}({1/\zeta})}$ 
for ${l' \neq l}$ and ${a_{l l} = 1}$\cite{bjornson2016deploying}. Note that $a_{l' l} = 1$ represents that all UEs in the cell $l'$ have the same pilot subset of those in the cell $l$ and $a_{l' l} = 0$  shows that the UEs in the cell $l$ and $l'$ use different pilot subsets\footnote{In this paper, the scheme in which only some UEs of cell $l$ and $l'$ use the same pilot subset is not considered.}. The pilot-sharing UEs in the first case leads to the pilot contamination which is a known problem in the massive MIMO network.

\begin{table}
\centering
\caption{Dual-slope path loss model.} 
\begin{tabular}{l|l|l}
   $r_n$ & $\Upsilon_n$ & $\alpha_n$ \\
\hline
$r_1=100$m & $\Upsilon_1= 8.3e-04 $, $\Upsilon_2= 5.2481$& $\alpha_1=2.1$, $\alpha_2=4$ 
\end{tabular}\label{table1}
\end{table}



\subsection{Power control policy and channel estimation}\label{subsec:powercontrol}
We consider a statistical channel inversion as power control policy such that the UE $i$ in the cell $l$ has the transmit power 
\begin{equation} \label{ULpower}
p_{li} = \rho_0/{\beta_{li}^{l}},
\end{equation}
where $\rho_0$ is a design parameter \cite{bjornson2015massive}. This ensures a uniform signal-to-noise ratio (SNR) at the BS as ${\mathsf{SNR}_0 = \rho_0/\sigma^2}.$ 
The same power control policy is used for pilot transmission phase with transmit power of the UE $i$ in the cell $l$ as $p^{\rm p}_{li} = \rho_{\tr}/{\beta_{li}^{l}}$ and $\rho_{\tr} > \rho_0$ is a design parameter. Then the SNR at the BS in the pilot phase is ${\mathsf{SNR}_{\tr} = \rho_{\tr}/\sigma^2}$. 

During the pilot transmission phase, the received signal $\vect{Y}^{\rm p}_{j} \in \mathbb{C}^{M \times \tau_p}$ at the BS $j$ is
\begin{eqnarray}\label{receivedpilot1}
\vect{Y}^{\rm p}_{j} =  \sum\limits_{i= 1}^K \sqrt{p^{\rm p}_{jk}} \vect{h}_{ji}^{j} \boldsymbol{\phi}_{jk}^{T} +  \sum_{l\in {\vect{\Psi}}_{\lambda} \setminus {\{j\}}} \sum\limits_{i = 1}^K \sqrt{p^{\rm p}_{li}} \vect{h}_{li}^{j} \boldsymbol{\phi}_{li}^{T} + \mathbf{N}_{j},
\end{eqnarray}
where $\mathbf{N}_{j} \in \mathbb{C}^{M \times \tau_p}$ is the additive receiver noise with i.i.d. elements distributed as $\mathcal{N}_{\mathbb{C}}(0, \sigma^2)$. 
The linear MMSE
estimate of $\vect{h}_{li}^{j}$ is obtained by  $\vect{Y}^{\rm p}_{j} \boldsymbol{\phi}_{li}^{*}$ as follows  \cite{long2018interference} 
\begin{eqnarray}\label{channelEst}
\widehat{\vect{h}}_{li}^{j} =  \sqrt{ p^{\rm p}_{li}} \vect{R}_{li}^{j} (\vect{Q}_{li}^{j})^{-1}\left( \vect{Y}^{\rm p}_{j} \boldsymbol{\phi}_{li}^{*}\right) 
\end{eqnarray}
where $\widehat{\vect{h}}_{li}^{j}  \sim  \mathcal{N}_{\mathbb{C}}(\mathbf{0}, p^{\rm p}_{li} \tau_p \vect{R}_{li}^j  (\vect{Q}_{li}^{j})^{-1} \vect{R}_{li}^j) $ and $\vect{Q}_{li}^{j} = \sum\nolimits_{l' \in \vect{\Psi}_\lambda }  a_{l'l} p^{\rm p}_{l'i} \tau_p \vect{R}_{l'i}^{j} + \frac{1}{\mathsf{SNR}_{\tr}} \vect{I}_M$. {We must note that this is the multi-cell MMSE channel estimator which contains the received signals of the UEs in cell $l$ and also the other cells with the same pilots. The single-cell MMSE channel estimator does not consider the second term and performs well only in a single-cell scenario.}
The estimation error $\tilde{\mathbf{h}}_{li}^{j} = {\mathbf{h}}_{li}^{j}  - \widehat{\mathbf{h}}_{li}^{j} $ has a correlation matrix 
\begin{eqnarray}\label{correlation_error}
\vect{C}_{li}^j = \vect{R}_{li}^j - p^{\rm p}_{li} \tau_p \vect{R}_{li}^j  (\vect{Q}_{li}^{j})^{-1} \vect{R}_{li}^j.
\end{eqnarray}
As we observe from \eqref{channelEst}, the pilot contamination is the result of the mutual interference generated by the pilot-sharing UEs and leads to two main consequences in the channel estimation process.  First, it  reduces the channel estimation accuracy, and second, the estimates $\{\widehat{\mathbf{h}}_{li}^{j}; l \in \vect{\Psi}_\lambda \}$ become correlated:
\begin{align}\label{theta_jmni}
\mathbb{E}\{ \widehat{\mathbf{h}}_{li}^{j} {({\widehat{\mathbf{h}}}_{l'i}^{j})}^{\Htran}\} =  \vect{\Phi}_{jl'li} = \frac{\sqrt{p^{\rm p}_{li}}}{\sqrt{p^{\rm p}_{l'i}}}\vect{R}^j_{li} (\vect{Q}^{j}_{li})^{-1} \vect{R}^j_{l'i}. 
\end{align}
Both consequences degrade the UEs' performance; however it is 	the second one that is responsible for the so-called \emph{coherent interference}, which may increase linearly with $M$ like the signal term. This will be investigated in Section \ref{sec:net_analyz}  in detail.

\begin{corollary}
In the special case of uncorrelated Rayleigh fading,  $\widehat{\mathbf{h}}_{jk}^j$  is simplified as follows:
\begin{eqnarray}\label{channelEst_Uncorr}
\widehat{\mathbf{h}}_{jk}^j = \gamma_{jk}^{j} \frac{ 1}{ \tau_p \sqrt{\beta_{jk}^{j} \rho_{\tr}}} \mathbf{y}_{jk}^{j}\sim \mathcal{N}_{\mathbb{C}}(\mathbf{0}, \gamma_{jk}^{j} \mathbf{I}_M),
\end{eqnarray}
where $\mathbf{y}_{jk}^{j} = \vect{Y}^{\rm p}_{j} \boldsymbol{\phi}_{jk}^{*}$ and
\begin{equation}\label{eq:sec4_delta_1}
\gamma_{jk}^{j} =  \beta_{jk}^{j}  \left(1 + \sum\limits_{l \in \vect{\Psi}_\lambda \setminus \{j\} }  a_{lj}     \frac{\beta_{l k}^{j}}{\beta_{l k}^{l}} +  \frac{1}{\tau_p \mathsf{SNR}_{\tr}}\right)^{-1}.
\end{equation} 
\end{corollary}



 \section{Network analysis}\label{sec:net_analyz}
 Theoretical analysis is conducted for ``typical UE'', which is statistically representative for any other UE in the network \cite{Baccelli2008a}. Without loss of generality, we assume that the ``typical UE'' has an arbitrary index $k$ and is connected to an arbitrary BS $j$. 
We analyze the SE per UE as the BS density $\lambda$ increases. The estimated channel used in the combiners has some level of estimation error analyzed through the normalized MSE (NMSE) in terms of $\lambda$ 
in Section \ref{subsec:Channel estimation accuracy}.
When the channel is modeled as uncorrelated Rayleigh fading and the MR or ZF is employed as the combining scheme,  the theoretical analysis can be applied to compute a lower bound for the SE \cite{pizzo2018network} which will be explained in Section \ref{subsec:UatF}. Finally, we study the SE in the general case of correlated fading with both single-cell and multi-cell processing schemes in Section \ref{SE_Correlated}.

\subsection{Channel estimation accuracy}\label{subsec:Channel estimation accuracy}
Channel estimation accuracy is measured through the average NMSE, which is defined as follows \cite{fahim}:
\begin{eqnarray}\label{NMSE_General}
\nonumber
\mathsf{NMSE} & = & \mathbb{E}_{\{\bf h, d, a\}}\{\mathsf{NMSE}_{jk}\}  \triangleq \mathbb{E}_{\{\bf d, a\}}  \left\{ \frac{ \mathbb{E}_{{\{\bf h\}}} \{\| \vect{h}_{jk}^{j} - \widehat{\vect{h}}_{jk}^{j} \|^2\}}{ \mathbb{E}_{{\{\bf h\}}}\{\| \vect{h}_{jk}^{j} \|^2\}}  \right\}  \\ &=& \mathbb{E}_{\{\bf d, a\}} \left\{ \frac{\tr(\vect{C}_{jk}^j)}{\tr(\vect{R}_{jk}^j)} \right\},
\end{eqnarray}
where $\mathsf{NMSE}_{jk}$ is the NMSE of the UE $k$ in the cell $j$ and the expectations are computed with respect to channel realizations $\bf h$,  pilot allocations $\bf a$, and UE positions $\bf d$. This is a value between 0 (perfect estimation) and 1 (using $\mathbb{E}{\{{\bf{h}}_{jk}^{j}\}}$ as the estimate). 
From \eqref{channelEst}, we can conclude that the NMSE is affected by two factors. First the value of SNR in the pilot phase, ${\mathsf{SNR}_{\tr}}$ and second the pilot contamination which comes from the pilot-sharing UEs. 
The impact of the network densification is on the latter (pilot contamination term)
by causing more number of pilot-sharing UEs.
Using \eqref{correlation_error}, the NMSE in \eqref{NMSE_General} can be simplified as follows
\begin{align}\label{NMSE_Corr_Upperbound}
\mathsf{NMSE^{Corr}} 
 =
1 - \mathbb{E}_{\{\bf d \}} \left\{ \frac{\tr\left( p^{\rm p}_{jk} \tau_p \vect{R}_{jk}^j  \mathbb{E}_{\{\bf a\}}\{ (\vect{Q}_{jk}^{j})^{-1}\} \vect{R}_{jk}^j\right)}{\tr(\vect{R}_{jk}^j)}   \right\},
\end{align}
In the following we obtain an upper bound for NMSE in the uncorrelated fading channel.

\begin{corollary}\label{corollary_NMSE} 
In the uncorrelated Rayleigh fading: 
\begin{align} \label{NMSE_Uncorr}
\mathsf{NMSE^{Uncorr}} = 1 - \mathbb{E}_{{\{\bf d\}}} \left\{\frac{1}{\beta_{jk}^{j}}\mathbb{E}_{{\{\bf a\}}} \left\{\gamma_{jk}^{j}\right\}\right\},
\end{align}
where $\gamma_{jk}^{j}$ is given in \eqref{eq:sec4_delta_1}.
Then the NMSE for uncorrelated Rayleigh fading can be upper bounded as follows 
\begin{align}\label{NMSE_Upperbound}
\mathsf{NMSE^{Uncorr}} \le  1 - \frac{1}{A(\mu_1) }, 
\end{align}
where 
\begin{align}\label{A_(mu_1)}
A(\mu_1) = 1 + \frac{\mu_1}{\zeta} + \frac{1}{ \tau_p \mathsf{SNR}_{\tr}},
\end{align}
and  $\mu_1$ is given in \eqref{mu_kappa} by setting $\kappa = 1$ with
\begin{figure*}
\begin{equation}\label{mu_kappa}
\!\!\!\!\!\!\!\! \mu_\kappa = 2 \sum_{n = 1}^{N} \frac{\Gamma\left(2, \pi \lambda r_{n-1}^2\right) - \Gamma\left(2, \pi \lambda r_{n}^2\right)}{\kappa \alpha_n - 2} + \frac{2 c_n(\kappa)}{(\pi \lambda)^{\frac{\kappa \alpha_n}{2} -1}} 
\left(  \Gamma\left(1+\frac{\kappa \alpha_n}{2}, \pi \lambda r_{n-1}^2\right) -  \Gamma\left(1+\frac{\kappa \alpha_n}{2}, \pi \lambda r_{n}^2\right)\right), \; \kappa=1,2
\end{equation}
\hrule
\end{figure*}
\begin{equation} \label{cn_kappa}
 c_n(\kappa) = - \frac{r_n^{2 - \kappa \alpha_n}}{ \kappa \alpha_n - 2} \sum_{i = n + 1}^N \left(\frac{\Upsilon_i}{\Upsilon_n}\right)^\kappa \frac{r_{i-1}^{2 - \kappa \alpha_i} - r_i^{2 - \kappa \alpha_i}}{\kappa \alpha_i - 2}.
\end{equation}
\end{corollary}
\begin{proof}
The proof is given in Appendix A.
\end{proof}

\begin{lemma}\label{lemma_increasingNMSE}
The NMSE is a non-decreasing function of BS density $\lambda$ due to the same behavior of $\mu_{1}$ with $\lambda$.
\end{lemma}
\begin{proof}
The proof is given in Appendix B.
\end{proof}

\subsection{Spectral efficiency analysis of uncorrelated Rayleigh fading}\label{subsec:UatF}
A lower bound for the SE computation is use-and-then-forget (UatF) bound wherein the channel estimates are used for designing the receive combining vectors and then effectively ``forgotten'' before the signal detection \cite{bjornson2016deploying}.

\begin{theorem}[UatF bound\cite{bjornson2016deploying}]\label{TheoremUatF}
The UL SE of the typical UE is lower bounded by 
\begin{eqnarray}\label{SE_UatF}
\mathsf{SE}^{\rm UatF} = \left( 1 - \frac{K \zeta}{\tau_c}\right) \mathbb{E}_{\mathbf{d}} \{ \log_2\left(1 + \mathsf{SINR}_{jk}^{\rm UatF}\right)\},
\end{eqnarray}
where $\mathsf{SINR}_{jk}^{\rm UatF}$ is the conditioned SINR on the realization of UE position, given by \eqref{SINR_Uatf}  and $\vect{v}_{jk} \in \mathbb{C}^{M}$ is the combining vector for UE $k$ in cell $j$.
\end{theorem}

 \begin{figure*}
\centering
\begin{equation}\label{SINR_Uatf}
\mathsf{SINR}_{jk}^{\rm UatF} = \frac{p_{jk}|\mathbb{E}_{\{\mathbf{h}, \mathbf{a}\}}\{\vect{v}_{jk}^H {\mathbf{h}}_{jk}^{j}\}|^2}
{ \sum\limits_{l \in \vect{\Psi}_{\lambda}} \sum\limits_{{i = 1 }}^K p_{li} \mathbb{E}_{\{\vect{h}, \mathbf{a}\}}\{|\vect{v}_{jk}^H {\vect{h}}_{li}^{j}|^2\} - p_{jk}|\mathbb{E}_{\{\vect{h}, \mathbf{a}\}}\{\vect{v}_{jk}^H {\vect{h}}_{jk}^{j}\}|^2 + \sigma^2 \mathbb{E}_{\{\vect{h}, \mathbf{a}\}}\{||\vect{v}_{jk}||^2\}}. 
\end{equation}
\hrule
\end{figure*}


In the case of MR and ZF combining schemes, we have  
\begin{equation} \label{MR_ZF}
\vect{V}_{j} \triangleq \left[ \vect{v}_{j 1} \, \ldots \, \vect{v}_{j K}  \right] = \begin{cases}
 \widehat{\vect{H}}_{j}^{j} & \textrm{with MR} \\ 
\widehat{\vect{H}}_{j}^{j} \left(
 (\widehat{\vect{H}}_{j}^{j})^{\Htran} \widehat{\vect{H}}_{j}^{j}  \right)^{-1} & \textrm{with ZF}
\end{cases}
\end{equation}
where $\widehat{\vect{H}}_{j}^{j} = [\widehat{\vect{h}}_{j1}^{j} \, \ldots \,  \widehat{\vect{h}}_{jK}^{j}] \in \mathbb{C}^{M \times K}$ is the matrix of channel estimates from all UEs in the cell $j$ to the BS $j$. Then, the UatF bound leads to the tractable bounds shown in \eqref{snr_MR} and \eqref{snr_ZF} for the MR and ZF combiners, respectively \cite{pizzo2018network}. Clearly, the first term in the denominator is related to noise, the second and the third terms are the interference decomposed into intra-cell and inter-cell ones. 
The last term which accounts for the pilot contamination is the same for both schemes. The interplay of these terms with the network parameters can be concluded from these bounds as follows.
 
 \begin{figure*}
\begin{align}\label{snr_MR}
\mathsf{SINR}^{\rm UatF, MR} &= \frac{1}{\underbrace{\frac{A(\mu_1)}{M \mathsf{SNR}_0}}_{\textnormal{Noise}} + \underbrace{ \frac{K}{M}A(\mu_1)}_{\textnormal{ Intra-cell \,Interference}} + \underbrace{\frac{K}{M}\left(A(\mu_1)\mu_1 + \frac{\mu_2}{\zeta}\right)}_{\textnormal{Inter-cell \, Interference}} + \underbrace{\frac{\mu_2}{\zeta}}_{\textnormal{Pilot\, Contamination}}}\\\label{snr_ZF}
\bigskip
\mathsf{SINR}^{\rm UatF, ZF} &= \frac{1}{\underbrace{\frac{A(\mu_1)}{(M-K) \mathsf{SNR}_0}}_{\textnormal{Noise}} + \underbrace{ \frac{K}{M-K}(A(\mu_1)-1)}_{\textnormal{Intra-cell \,Interference}} + \underbrace{\frac{K}{M-K}A(\mu_1)\mu_1}_{\textnormal{Inter-cell \, Interference}} +\underbrace{\frac{\mu_2}{\zeta}}_{\textnormal{Pilot\, Contamination}}}
\end{align}
\hrule
\end{figure*}

\textbf{Impact of $M$ and $K$}:
The terms related to the noise and interference, decrease linearly with $M$ for MR and decrease with ${M - K}$ for ZF. Because ZF aims to suppress the intra-cell interference by sacrificing $K$ spatial dimensions.
The sum of the two terms is generally referred to as non-coherent interference due to their relation with $M$. 
Besides, the same happens for the noise. The pilot contamination is independent of $M/K$ and hence it is also called coherent interference.

\begin{corollary}\label{corollary_AsymSE}
{When ${M\to \infty}, \mathsf{SINR}^{\rm UatF, MR} =\mathsf{SINR}^{\rm UatF, ZF} \xrightarrow[M\to \infty]{\rm a.s.} \gamma_{M\to \infty} = {\zeta}/{\mu_2}$ 
and the ultimately achievable rate is} 
\begin{align}\label{R_inf}
R_{M\to \infty} = \left(1 - \frac{\zeta K}{\tau_c}\right)\log_2\left(1 + \frac{\zeta}{\mu_2}\right).
\end{align}
where $\xrightarrow[M\to \infty]{\rm a.s.}$ denotes almost sure convergence. 
\end{corollary} 

Then, we can derive the optimal pilot reuse factor for the asymptotic case as follows: 
\begin{corollary}\label{corollary_optimalPRF} 
{When ${M\to \infty}$, the optimal pilot reuse factor $\zeta$ that maximizes \eqref{R_inf} is}
\begin{align}\label{zeta_inf}
\zeta_{M\to \infty}^{\mathsf{opt}} = \mu_2 \left( \frac{\nu}{W(\nu e)} - 1\right),
\end{align}
where $W(.)$ is Lambert function and $\nu = 1 + {\tau_c}/({\mu_2 K})$.
\end{corollary} 
\begin{proof}
The proof is given in Appendix C.
\end{proof}

\textbf{Impact of imperfect CSI and $\lambda$}:
{Considering \eqref{snr_MR} and \eqref{snr_ZF}, the (inter and intra-cell) interference and noise are linearly increasing functions of ${A(\mu_1)}$ presented in \eqref{A_(mu_1)} which comes from the imperfect knowledge of the channel (Corollary \ref{corollary_NMSE}). 
 Moreover, ${A(\mu_1)}$ is in turn a function of  $\lambda$ and $\zeta$. Hence, the  intra-cell interference and noise are linearly increasing with $\lambda$ while the inter-cell interference is a function of ${A(\mu_1)}\mu_1$ and then increases at an even faster rate with $\lambda$ as is proved in Lemma \ref{lemma_increasingNMSE}. The pilot contamination is also a non-decreasing function of $\lambda$ (due to the behavior of $\mu_2$ with $\lambda$ proved in Appendix B). Clearly, both interference and pilot contamination depend inversely on $\zeta$.} 
 The following corollary is concluded. 

\begin{corollary}\label{corollary_ComparingReductionRate}
The SE of ZF has larger reduction rate than the SE of MR with respect to $\lambda$. In other words, the $\mathsf{SINR}^{\rm UatF, ZF}$ is decreasing faster than the $\mathsf{SINR}^{\rm UatF, MR}$ with $\lambda$.
\end{corollary}

\begin{proof}
The proof is given in Appendix D.
\end{proof}


\textbf{Coherent and non-coherent interference regions}:
It is interesting  to investigate that for which values of $M/K$, the pilot contamination is
dominant rather than the (inter- and intra-cell) interference and vice versa.
The following corollary answers this question.

\begin{corollary}\label{corollary_dominatingterm} 
The pilot contamination is the dominating term when $M > K ( 1 + \zeta A(\mu_1) \frac{1 + \mu_1}{\mu_2}) $ and 
$M > K ( 1 + \zeta A(\mu_1) \frac{1 + \mu_1}{\mu_2} - \frac{\zeta}{\mu_2}) $ for MR and ZF, respectively. 
\end{corollary}
The required value of $M$ to have larger pilot contamination (than intra- and inter-cell interference)
for ZF is smaller than the one for MR for any given values of $\lambda$ and $\zeta$. Because the ZF has lower non-coherent interference than the one for MR. Clearly, the required $M$ to define the pilot contamination region (the region where the pilot contamination is the dominating term)  is a non-decreasing function of $\lambda$.


\subsection{Spectral efficiency analysis of correlated Rayleigh fading}\label{SE_Correlated}
The UatF bound \eqref{SE_UatF} is held on the basis of the channel hardening property in massive MIMO. On the other hand, larger spatial channel correlation can decrease the level of channel hardening\cite{massivemimobook}. In the general spatially correlated channel a more exact lower bound on UL SE is given by the following theorem which holds for any combining scheme, UE positions and pilot allocations. 

\begin{theorem}[Average ergodic SE\cite{pizzo2018network}]\label{Theorem1}
If MMSE channel estimation \eqref{channelEst} is used, then the UL SE of the typical UE  is lower bounded by $\underline{\mathsf{SE}}$ as follows
\begin{eqnarray}\label{SE_Exact}
\underline{\mathsf{SE}} = \left( 1 - \frac{K \zeta}{\tau_c}\right) \mathbb{E}_{\{\mathbf{d}, \mathbf{a}, \vect{h}\}} \{ \log_2\left(1 + \mathsf{SINR}_{jk}\right)\},
\end{eqnarray}
where
\begin{align}
\mathsf{SINR}_{jk} = \frac{p_{jk}|\vect{v}_{jk}^H \widehat{\vect{h}}_{jk}^{j}|^2}
{ \sum\limits_{l \in \Psi_{\lambda}} \sum\limits_{\substack{i = 1 \\ (l, i) \neq (j, k)}}^K p_{li} |\vect{v}_{jk}^{\Htran} \widehat{\vect{h}}_{li}^{j}|^2 + \vect{v}_{jk}^{\Htran} {\bf{Z}}_j \vect{v}_{jk}},\label{SINR}
\end{align}
is the instantaneous SINR of the typical UE $k$ in cell $j$, 
$\vect{v}_{jk} \in \mathbb{C}^{M}$ is the combining vector, and the expectations is computed with respect to the channel
realizations, pilot allocations (showed in $\widehat{\vect{h}}_{li}^{j}$), and UE positions. Also
\begin{align}
\vect{Z}_{j}= \sum\limits_{l \in \vect{\Psi}_\lambda} \sum\limits_{i=1}^K  {p}_{li}^j\vect{C}_{li}^j  + \sigma^2 \mathbf{I}_M.
\end{align}
\end{theorem}


We must note that analytic computation of the expectations in \eqref{SE_Exact} is demanding and hence we apply the numerical analysis in this case.
We notice that the SINR in \eqref{SINR} only depends on one combining vector, ${\bf v}_{jk}$, and is the form of a generalized Rayleigh quotient \cite[Lemma B.10]{massivemimobook}. Hence, the combining vector that maximizes \eqref{SINR} can be obtained: 

\begin{corollary}
The instantaneous UL SINR in \eqref{SINR} for a typical UE $k$ in cell $j$ is maximized by a combining which is the $k^{\rm th}$ column of the following matrix \cite{massivemimobook}
\begin{align}\label{MMMSE}
\vect{V}_{j}^{\rm M-MMSE} = \left( \sum\limits_{l \in \vect{\Psi}_\lambda}  \widehat{\mathbf{H}}_{l}^{j}{\mathbf P}_l ( \widehat{\mathbf{H}}_{l}^{j} \big)^{\Htran} + \vect{Z}_{j} \right)^{-1}\widehat{\mathbf{H}}_{j}^{j}{\mathbf P}_j.
\end{align}
{where $\mathbf{P}_l$ is a $K \times K$ diagonal matrix containing the transmit powers of all UEs in cell $l$, $p_{li}, i = 1, \cdots, K$.}
\end{corollary}
This combining is called M-MMSE since \eqref{MMMSE} also minimizes ${\rm{MSE}}_{jk}^{\rm {ul}} = \mathbb{E} \{|s_{jk} - \vect{v}_{jk}^{\Htran}\vect{y}_{jk}^j|^2\,|\, \{ \widehat{\vect{h}}_{li}^j \}\}$.
However, due to high complexity of M-MMSE, the single-cell processing schemes such as ZF, MR (given by \eqref{MR_ZF} in Section \ref{subsec:UatF}) and S-MMSE are widely used in the literature \cite{hoydis2013massive}:
\begin{align}\label{SMMSE}
\vect{V}_{j}^{\rm S-MMSE}= \left( \widehat{\mathbf{H}}_{j}^{j}{\mathbf P}_l ( \widehat{\mathbf{H}}_{j}^{j} \big)^{\Htran} + {\bar{\vect{Z}}}_j \right)^{-1}\widehat{\mathbf{H}}_{j}^{j}{\mathbf P}_j .
\end{align}
with ${\bar{\vect{Z}}}_j =  \sum\limits_{i=1}^K  {p}_{ji}\vect{C}_{ji}^j  + \sum\limits_{l \in \vect{\Psi}_\lambda}\sum\limits_{i=1}^K  {p}_{ji}\vect{R}_{li}^j+ \sigma^2 \mathbf{I}_M$.


{The channel estimates $\{\hat{\vect{h}}_{ji}^j:i=1,\ldots,K\}$ in the Poisson-Voronoi region of the serving BS are computed in S-MMSE, whereas the M-MMSE utilizes also the channel estimates of the UEs in other cells, which is computed locally in the BS. Then, S-MMSE is weaker than M-MMSE in decreasing the interferences from the UEs in other cells, i.e., inter-cell interference.} The ZF has the ability to reject the interferences from the UEs located in the Poisson-Voronoi region of the serving BS, i.e., intra-cell interference and MR maximizes the desired signal power. {The computational complexity of all combining schemes is as follows. The M-MMSE, S-MMSE, ZF and MR have the computational complexity of $\mathcal{O}(M^3+\lambda KM^2)$, $\mathcal{O}(M^3+ KM^2)$, $\mathcal{O}( KM^2 + K^3)$ and $\mathcal{O}(KM^2)$, respectively\cite{massivemimobook}. Assuming $M \gg K$, the M-MMSE and S-MMSE have roughly the complexity of $\mathcal{O}(M^3)$.}


In order to observe the behavior of the SE with respect to the BS density for the combining schemes, we perform the analysis based on \eqref{SE_Exact} and illustrate them in Fig.~\ref{fig:SE_MMSEBound} with parameters $M/K = {10}$, $K = 10$, $\zeta = 4$. Although the analysis is valid for any path loss model, we choose the dual-slope path loss model reported in {Table \ref{table1}.}
The SE versus $\lambda$ for the correlated Rayleigh fading is illustrated in Fig.~\ref{figSE:sub-first} with $\Delta = 10^\circ$ and for the uncorrelated channel is shown in Fig.~\ref{figSE:sub-second}. 


We observe in both figures that the SE per each UE is a non-increasing function of $\lambda$ for all schemes. 
This comes from the fact that, as the network becomes densified, both (intra- and inter-cell) interferences  and pilot contamination are increasing, {whereas due to the given power control \eqref{ULpower}}, the BS experiences a uniform SNR (Section \ref{subsec:powercontrol}) 
across the network. Hence, the SINR and then the SE are non-increasing functions of $\lambda$ which are also proved theoretically for MR and ZF through previous section for the uncorrelated Rayleigh fading. Since the dual-slope path loss model is simplified to the single-slope with $\alpha = \alpha_1$ when $\lambda \rightarrow \infty$ and $\alpha = \alpha_2$ when $\lambda \rightarrow 0$, the SE per each UE is independent  of $\lambda$ in this single-slope regions\cite{fahim, bjornson2016deploying}. As we observe, the network with $\lambda > 200$ $[\rm BS/km^2]$ is the dense network with $\alpha = 2.1$ and the network with $\lambda < 10$ $[\rm BS/km^2]$ is the low dense with $\alpha = 4$. 

\begin{figure*}[!t]
\centerline{\subfloat[]{\includegraphics[width=.55\columnwidth]{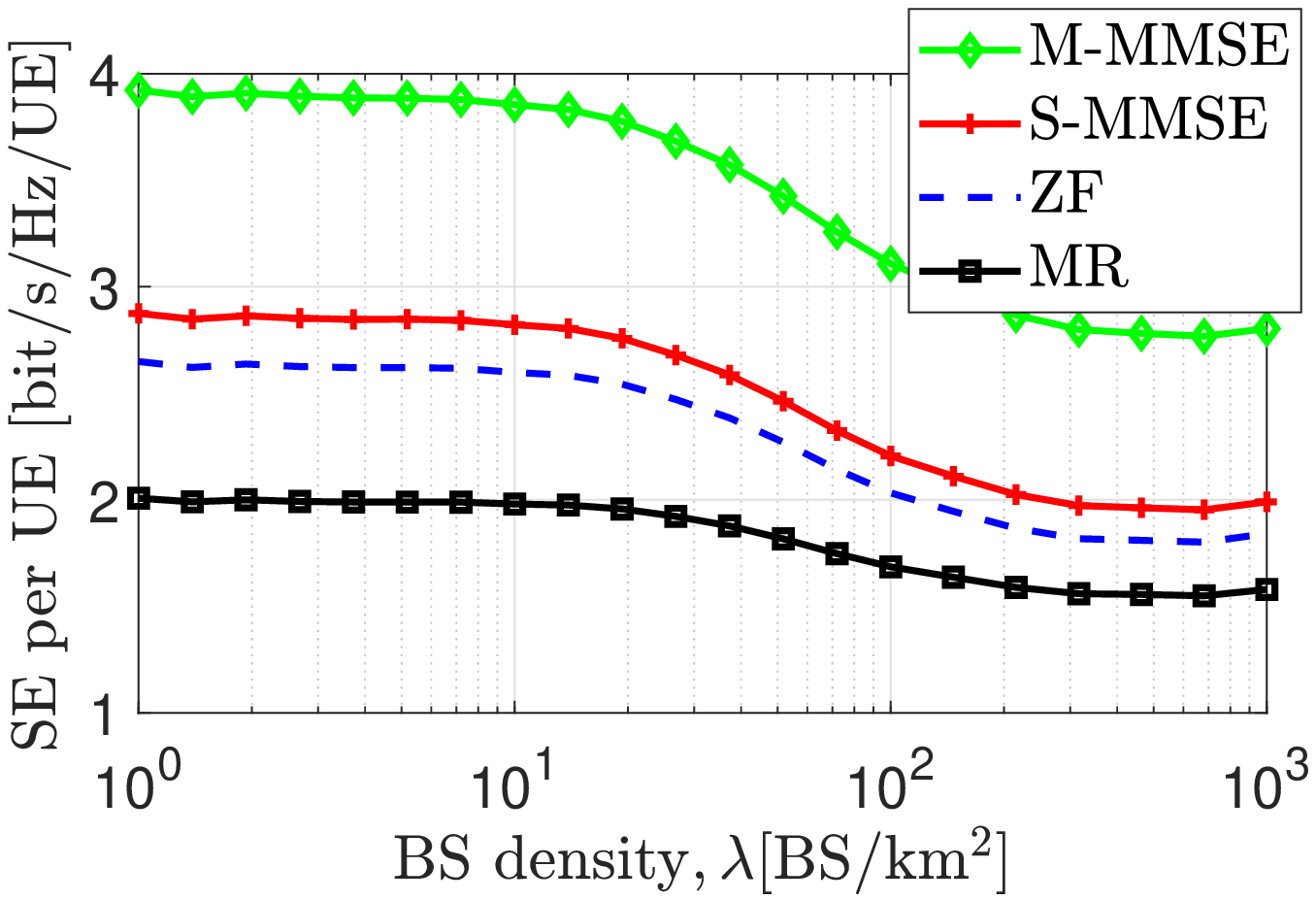}
\label{figSE:sub-first}}
\hfil
\subfloat[]{\includegraphics[width=.55\columnwidth]{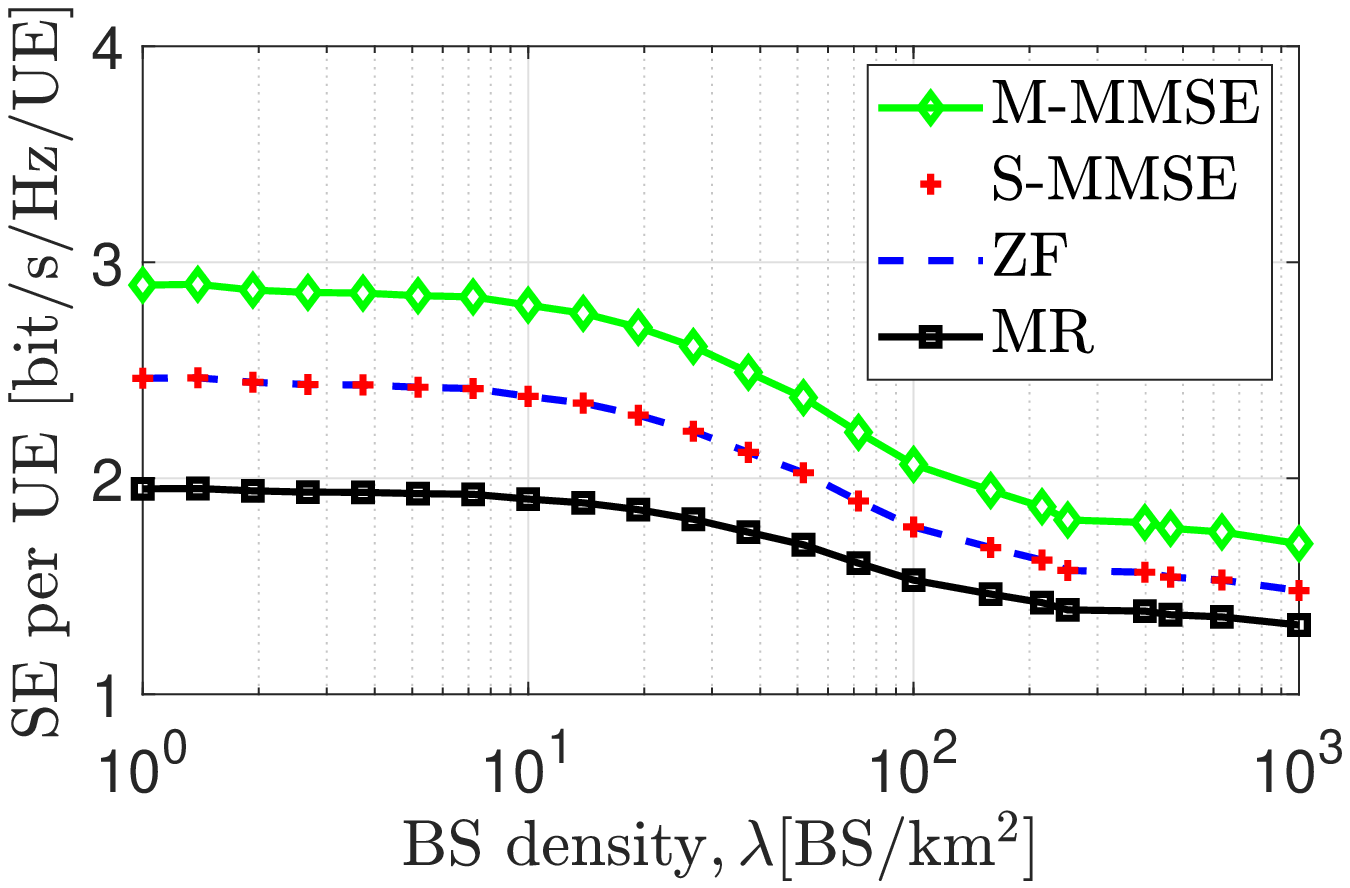}
\label{figSE:sub-second}}}
\caption{SE as a function of $\lambda$ $[\rm{BS/km^2}]$  based on lower bound given by \eqref{SE_Exact} for given dual-slope path loss model in Table \ref{table1}  with $M/K = {10}$, $K = 10$, $\mathsf{SNR}_0 = 5$ dB and $\zeta = 4$ for correlated Rayleigh fading with $\Delta = 10^\circ$  in (a) and uncorrelated Rayleigh fading in (b).} 
{\label{fig:SE_MMSEBound}}
\end{figure*} 
 
In addition, we observe from Fig.~\ref{fig:SE_MMSEBound} that each combiner scheme in the correlated fading provides higher SE than the related one for the uncorrelated in all BS density values. {This comes from 
that the channel estimation accuracy is higher in the correlated fading than the uncorrelated one for all different values of $\lambda$ which will also be shown in Fig. \ref{fig:NMSE} of Section \ref{sec:simulationresults}.}
Moreover, we demonstrate in Fig.~\ref{fig:SE_MMSEBound} that the multi-cell processing, M-MMSE holds high SE in all BS density values, since it is designed to suppress both intra-cell and inter-cell interference.
 As we observe, in the uncorrelated fading, as $\lambda$ increases and both the inter-cell interference and the pilot contamination become larger, the gap between the single-cell and multi-cell combining schemes gets smaller. This means that the M-MMSE ability in reducing the interference is not substantial in the dense network when the channel is uncorrelated. However,  in the correlated fading, the reduction rate of the SE obtained by the M-MMSE with respect to $\lambda$ is not as much as the one for the uncorrelated case. 
{Because, considering \eqref{channelEst}, we have $\widehat{\vect{h}}_{jk}^j - c \widehat{\vect{h}}_{lk}^{j} = ( \sqrt{ p_{jk}} \vect{R}_{jk}^j - c \sqrt{ p_{lk}}\vect{R}_{lk}^j) \mathfrak{t}_{jk}$ where $\mathfrak{t}_{jk} =  \vect{R}_{jk}^{j} (\vect{Q}_{lk}^{j})^{-1}\left( \vect{Y}^{\rm p}_{j} \boldsymbol{\phi}_{lk}^{*}\right)$. If $\vect{R}_{jk}^{j}$  and $\vect{R}_{lk}^{j}$ are not linearly  dependent, there are not any values of $c$ which result in $\widehat{\vect{h}}_{jk}^j - c \widehat{\vect{h}}_{lk}^{j} = 0$. This non-zero difference means that the channel estimates are linearly independent which is more probable to happen in the real channel. Hence, we can design a combiner scheme which is not only orthogonal to the channel of interfering UE, $\widehat{\vect{h}}_{lk}^j$, but also has a non-zero inner product with its desired channel $\widehat{\vect{h}}_{jk}^j$ \cite{BHS18A}.}
When the network becomes densified, the number of pilot-sharing UEs are also increasing, then the M-MMSE, as an optimum combiner, is able to reject/reduce the interference from the pilot-sharing UEs. 
In the S-MMSE, the average channel information from the UEs in other cells is considered, and hence the correlation feature also leads to the observation that S-MMSE has a bit better performance than the ZF in the correlated channel. Although when the channel is uncorrelated, the S-MMSE shows roughly the same SE as ZF. 

{In addition, the Fig.~\ref{fig:SE_MMSEBound} represents that the reduction rate of the SE obtained by MR with respect to $\lambda$ is lower than other combiners (Corollary \ref{corollary_ComparingReductionRate}). Since MR aims to maximize the desired received signal power, and it is not supposed to reduce the interference; however, other combiners sacrifice some numbers of spatial dimensions to overcome the interference. Then, in the dense network with much higher interference, this appears as larger reduction rate of the SE with respect to $\lambda$.} 

{Table \ref{table2} shows the variation of the angular spread on the SE and also ASE in which a low dense and a dense network are considered with $\lambda = 10 \ [\rm{BS/km^2}]$ and $\lambda = 50\ [\rm{BS/km^2}]$, respectively and  {$\Delta = \{0^\circ, 5^\circ, 10^\circ\}$.} We observe that higher SE and hence ASE are achieved in the channel with high correlation (smaller value of $\Delta$) by all combining schemes, because the channel with high correlation has lower NMSE.} 
Moreover, it is interesting to consider ASE $[{\rm bit/sec/Hz/km^2}]$  of the network which is given by ${\rm ASE} = \lambda . K . {\rm SE}$.
{The ASE for both correlated and uncorrelated channels are illustrated in  Fig. \ref{fig:ASE_MMSEBound}.
Although  the SE per UE has a non-increasing behavior versus $\lambda$, we observe that ASE increases with the BS density. Because the coefficient $\lambda . K$ in ASE compensates the effect of the multi slope path loss shown in SE.}


\begin{table*}[t]
\centering
\caption{{SE and ASE for correlated Rayleigh fading with angular spread, $\Delta$}} 
\begin{tabular}{|c||c|c|c|c|c|c||c|c|c|c|c|c|}
\hline
 & \multicolumn{6}{|c||}{$\lambda = 10$ $[{\rm BS / km^2}]$} & \multicolumn{6}{|c|}{ $\lambda = 50$ $[{\rm BS / km^2}]$} \\ 
\hline
 & \multicolumn{2}{|c|}{$\Delta = 0^\circ$} & \multicolumn{2}{|c|}{$\Delta = 5^\circ$} & \multicolumn{2}{|c||}{$\Delta = 10^\circ$} & \multicolumn{2}{|c|}{$\Delta = 0^\circ$} & \multicolumn{2}{|c|}{$\Delta = 5^\circ$} & \multicolumn{2}{|c|}{$\Delta = 10^\circ$} \\
 \hline
Scheme &  SE & ASE  & SE & ASE & SE  & ASE   &  SE & ASE  & SE & ASE & SE  & ASE \\
\hline   \hline
MR &  2.93 & 293 &  2.20 & 220 &  1.98 & 198 & 2.67 & 1335  & 1.95 & 975 & 1.80 & 900 \\ 
\hline
ZF & 3.39 & 339 &  2.87 & 287 & 2.60 & 260 & 2.99 & 1495  & 2.40 & 1200 & 2.25 & 1125 \\
\hline
S-MMSE & 5.21 & 521 & 3.14 & 314 & 2.84 & 284 & 4.88 & 2440 & 2.63 & 1315 & 2.44 & 1220 \\
\hline
M-MMSE & 5.33 & 533 & 4.4 & 440 & 3.86 & 386 & 5.02 & 2510 & 3.86 & 1930 & 3.39 & 1695 \\
\hline
\end{tabular}\label{table2}
\end{table*}

\begin{figure*}[!t]
\centerline{\subfloat[]{\includegraphics[width=0.55\columnwidth]{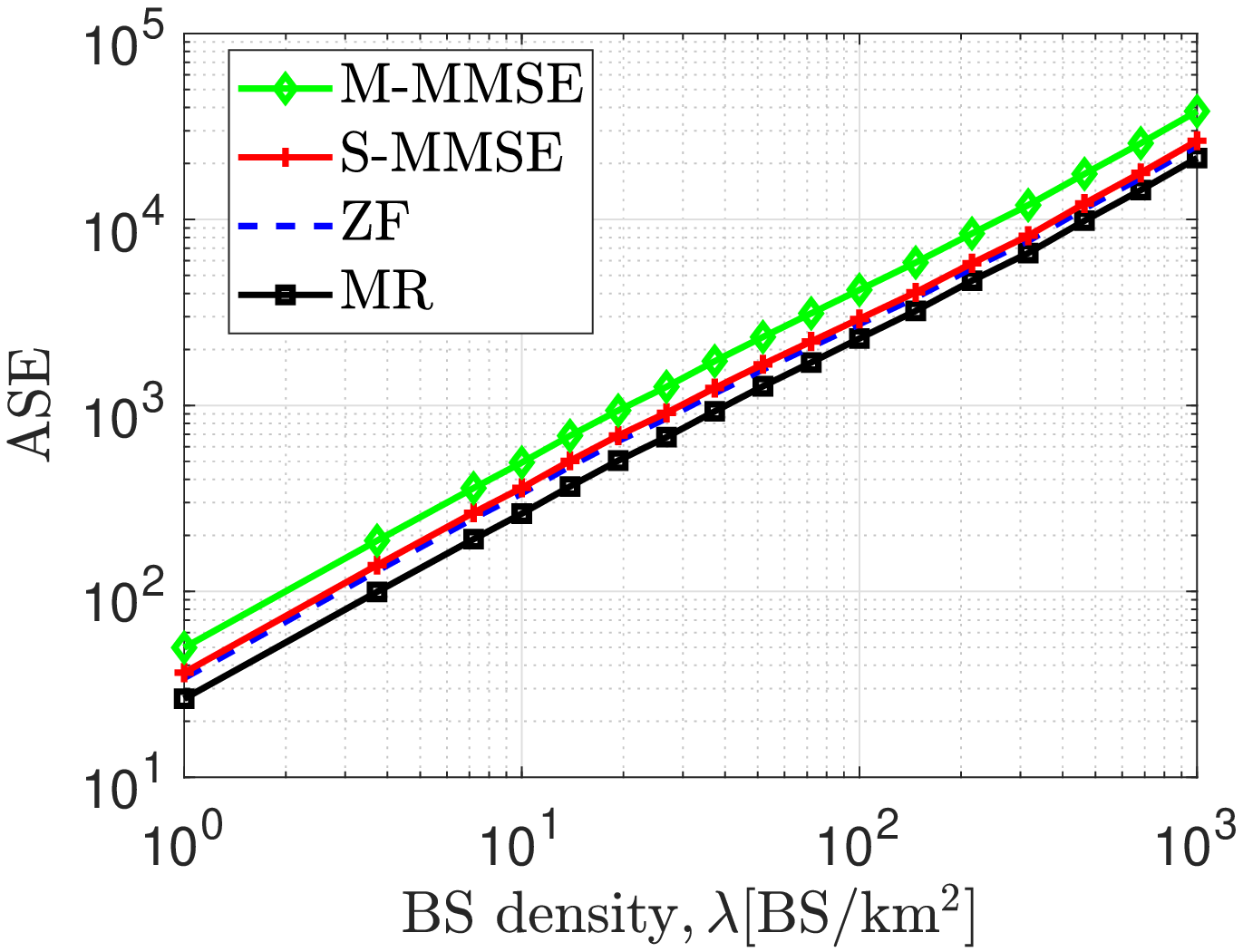}
\label{figASE:sub-first}}
\hfil
\subfloat[]{\includegraphics[width=0.55\columnwidth]{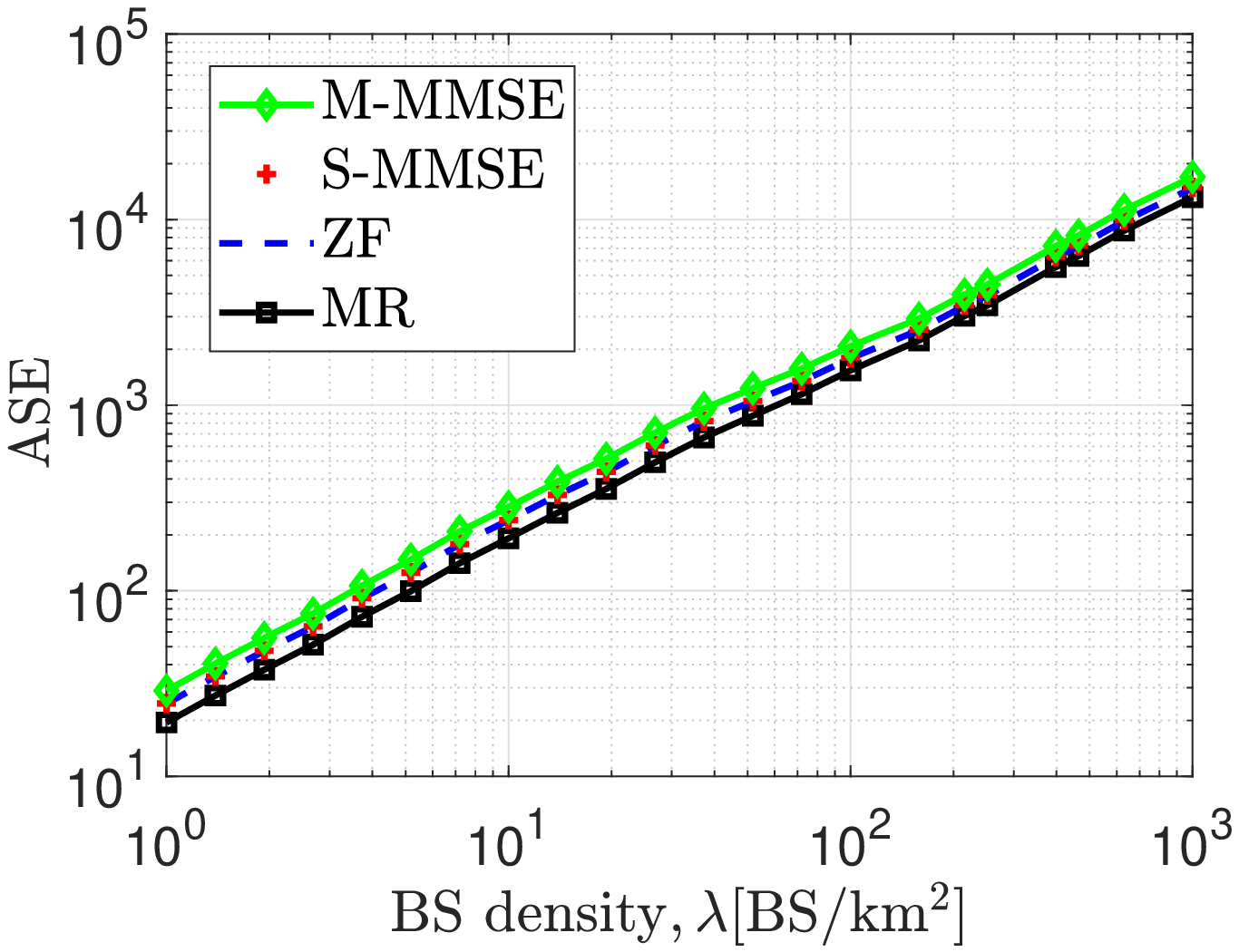}
\label{figASE:sub-second}}}
\caption{ASE as a function of $\lambda$ $[\rm{BS/km^2}]$  based on lower bound given by \eqref{SE_Exact} for given dual-slope path loss model in Table \ref{table1}  with $M/K = {10}$, $K = 10$, $\mathsf{SNR}_0 = 5$ dB and $\zeta = 4$ for correlated Rayleigh fading with $\Delta = 10^\circ$ in (a) and uncorrelated Rayleigh fading in (b) .} {\label{fig:ASE_MMSEBound}}
\end{figure*}

\section{Simulation results}\label{sec:simulationresults}
{In this section we illustrate the network performance based on the network parameters, $\lambda$, $M/K$, $\zeta$ and $\Delta$ by considering the path loss model given in {Table \ref{table1}.} We have the explanatory paragraph at the beginning of each figure which states the specified values of the network parameters used in it.}


\subsection{Channel estimation accuracy}
{Fig.~\ref{fig:NMSE} illustrates the NMSE for a correlated channel as a function of the BS density $\lambda$ when the angular spread, $\Delta = 5^\circ$ and $10^\circ$ and the number of antennas, $M = 100$. Two pilot reuse factors $\zeta = 1$ and $4$ are considered. The NMSE for the uncorrelated channel is also shown for comparisons. In a network with higher BS density, the estimated channel experiences larger interference from pilot-sharing UEs which causes more NMSE. Moreover, the channel with high correlation (smaller value of $\Delta$) has smaller NMSE. 
Intuitively, larger eigenvalues of the correlation matrix represent that there are some strong eigendirections i.e., with higher SNR which can be estimated independently with lower estimation error. {In the case of the uncorrelated fading, all the eigenvalues are the same, while in the correlated one, the eigenvalue decomposition spreads into different values wherein larger ones improve the NMSE. The eigenvalues of the typical correlation matrix are shown in Fig.~\ref{fig:Eigenvalue} for the correlated channel with $\Delta = 5^\circ$, $\Delta = 10^\circ$, and also the uncorrelated channel.}
With high spatial correlation (i.e., $\Delta = 5^\circ$), the NMSE is marginally affected by the BS density in the low dense and dense networks. An angular spread of $\Delta = 10^\circ$ is already enough to make the NMSE increase rapidly with $\lambda$, especially with 
$\zeta = 1$. To further reduce the NMSE, one must increase $\zeta$.}

\begin{figure}
\centering
\includegraphics[width=0.6\textwidth]{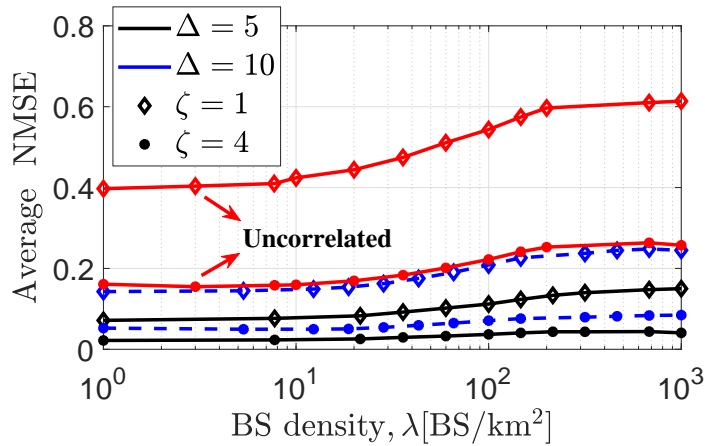}
\caption{Average NMSE as a function of $\lambda$ $[\rm{BS/km^2}]$ in the correlated Rayleigh fading channel with angular spread, $\Delta = \{5^\circ, 10^\circ\}$, pilot reuse factor $\zeta = \{1, 4\}$, $M = 100$ and given dual-slope path loss model in Table \ref{table1}. The uncorrelated case is also obtained.}\label{fig:NMSE}
\end{figure}

\begin{figure}
\centering
\includegraphics[width=0.6\textwidth]{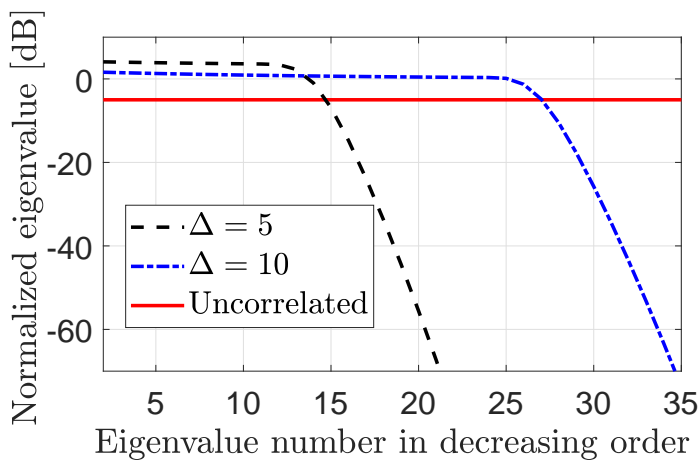}
\caption{{Eigenvalues of the spatial correlation matrix in the correlated Rayleigh fading channel with angular spread, $\Delta = \{5^\circ, 10^\circ\}$.}}\label{fig:Eigenvalue}
\end{figure}

\subsection{Spatially uncorrelated fading}
We consider the UatF bound in \eqref{SE_UatF} simplified to \eqref{snr_MR} and \eqref{snr_ZF} and illustrate the SE versus $\lambda$ for MR and ZF through Fig.~\ref{fig:SE_UatF}. We consider not only the practical antenna-UE ratio ${M/K=10}$  but also ${M/K=50}$ (for comparison with the asymptotic case). The latter case is studied to evaluate the gap with the ultimately achievable rate $R_{M\to \infty}$ in \eqref{R_inf}, shown in Fig.~\ref{fig:SE_UatF} with the optimal pilot reuse factor derived in Corollary \ref{corollary_optimalPRF}. However, all the other curves use the optimal pilot reuse factor obtained by exhaustive search. {A simple computation from \eqref{zeta_inf} shows that the optimal $\zeta$ for the asymptotic case is $5$ for $\lambda \leq 30 $ and $6$ for $\lambda > 30 $. While for the finite ratio of $M/K$, these are roughly $4$  and $5$ for $\lambda \leq 30$  and $\lambda > 30\ [{\rm BS/km^2}]$, respectively}. 
 Although, the ZF outperforms MR for $\lambda \leq 10$, the gain reduces as $\lambda$
increases. Both schemes achieve the same performance for
$\lambda > 30$. This is because the inter-cell interference becomes the
dominating term in \eqref{snr_MR} and \eqref{snr_ZF} when the BSs are much
closer to each other. 
We demonstrate  that both MR and ZF provide low SE especially at larger $\lambda$, even when $M/K$ takes large values (up to 50). 

\begin{figure}
\centering
\includegraphics[width=0.6\textwidth]{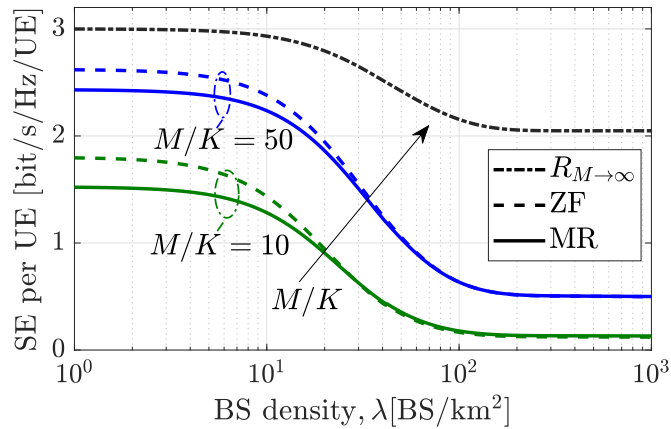}
\caption{SE as a function of $\lambda$ $[\rm{BS/km^2}]$ based on UatF bound for the ZF and MR in the uncorrelated Rayleigh fading channel. The dual-slope path loss model in Table \ref{table1} is adopted, $M/K = \{10, 50\}$, $K = 10$.}\label{fig:SE_UatF}
\end{figure}

\begin{figure}
\centering
\includegraphics[width=0.6\textwidth]{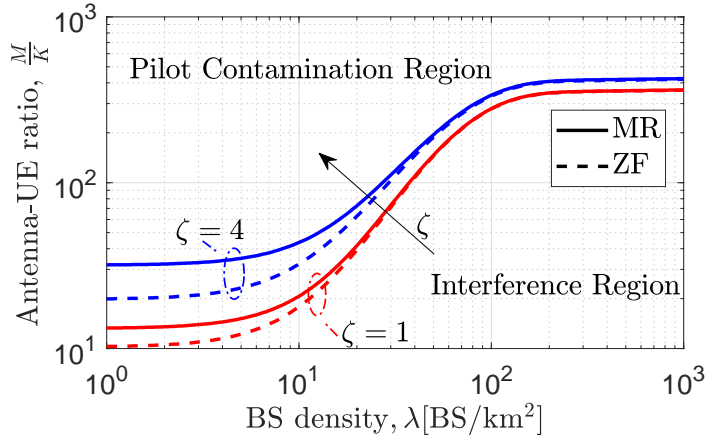}
 \caption{Antenna-UE ratio $M/K$ needed for the interference to equal pilot contamination as a function of $\lambda$. The dual-slope path loss model in Table \ref{table1} is adopted. Also, $\zeta = \{1,4\}$ 
 and $K=10$.}
\label{fig:Region}
\end{figure}

Fig.~\ref{fig:Region} illustrates the antenna-UE ratio $M/K$ for which the non-coherent interference and pilot contamination in \eqref{snr_MR} and \eqref{snr_ZF} are equal to each other obtained in Corollary \ref{corollary_dominatingterm}. The curves are plotted as a function of $\lambda$ for two values of pilot reuse factor ${\zeta\in \{1,4\}}$ and must be understood in the sense that the pilot contamination is dominant in the region above each one. As expected, MR has smaller pilot contamination region than ZF for any values of the BS density.  However, as the network becomes denser, the gap between the MR and ZF gets smaller since the inter-cell interference would be dominant.
{Remarkably, we observe that for practical values of $M/K$ in the interval ${4\le M/K \le 10}$,
regardless of $\lambda$ and $\zeta$, the pilot contamination region is never experienced, i.e., the non-coherent interference is the major interference.}
However, for larger values of $M/K$, $10 \le M/K \le 100$, the non-coherent interference always
dominates the pilot contamination in $\lambda \ge 50$. Hence, we conclude from Fig.~\ref{fig:Region} that in the dense network, the interference in massive MIMO plays a major role and the pilot contamination is dominant only for impractical values of $M/K$.

\begin{table*}[t]
\centering
\caption{{Average UL power of the non-coherent interference and coherent interference ([dB]) for the typical UE for both uncorrelated and correlated Rayleigh fading with $\Delta = 10^{\circ}$, $M = 100$, $K = 10$ and $\zeta  = 4$}} 
\begin{tabular}{|c||c|c||c|c|}
\hline
 & \multicolumn{2}{|c||}{$\lambda = 10$ $[{\rm BS / km^2}]$} & \multicolumn{2}{|c|}{ $\lambda = 50$ $[{\rm BS / km^2}]$} \\ 
\hline
Scheme& Uncorr.  & Corr.  & Uncorr.  & Corr. \\
&  (non-coherent / coherent) &  (non-coherent / coherent) &  (non-coherent / coherent) &  (non-coherent / coherent) \\
\hline   \hline
MR &  27.35 / 21.68 & 30.4 / 13.93   &   27.96 / 22.88 & 31 / 15.16   \\ 
\hline
ZF &   24.60 / 21.31 &  25.91 / 13.11  &   25.73 / 22.52  & 27.06 / 14.38 \\
\hline
M-MMSE & 22.31 / 21.07 & 17 / 6.27  & 23.62 / 22.27 & 18.89 / 7.95  \\
\hline
\end{tabular}\label{table3}
\end{table*}


\subsection{Impact of spatial channel correlation}
{As we mention, angular spread is the parameter of spatial channel correlation and show how correlated the channel is.  
The impact of correlation on coherent (pilot contamination) and non-coherent interference reduction is given by Table \ref{table3} for $\Delta = 10^\circ$, $M = 100$, $K = 10$, $\zeta = 4$ and also both low dense and dense network with $\lambda = 10\ [{\rm BS/km^2}]$ and $\lambda = 50\ [{\rm BS/km^2}]$, respectively. Although, the coherent interference of all schemes are comparable for uncorrelated channel, the spatial correlation leads to significant reduction of coherent interference, about 7 dB for MR, 8 dB for ZF and 14 dB for M-MMSE. 
Besides, as we explained in Section \ref{SE_Correlated}, the M-MMSE is able to decrease the non-coherent interference. This table also shows that dominant region in correlated Rayleigh fading for practical values of $M / K$ is non-coherent interference.}
Moreover, we investigate the tradeoff between angular spread and antenna-UE ratio $M/K$ to achieve the determined SE. In order to achieve the  SE = 3 $[\rm bit/s/Hz/UE]$, in the dense network with $\lambda = 50$ $[\rm BS/km^2]$, the required $M/K$ in terms of angular spread is shown in Fig.~\ref{fig:M/KvsASD}. Clearly, M-MMSE needs lower $M/K$ to reach the goal. Moreover, as the channel becomes correlated, we need lower $M/K$ to achieve the determined performance. On the other hand, the required $M/K$ for ZF and S-MMSE converge to each other at larger angular spread.
{We also illustrate the ASE in terms of $M/K$ in Fig.~\ref{fig:ASEvsM/K} for the correlated Rayleigh fading with $\Delta = 10^\circ$ and a dense network with $\lambda = 50\ [{\rm BS/km^2}]$ considering different values of $\zeta = \{1, 2, 4\}$. As we observe, the M-MMSE scheme with $\zeta = 2$ or $4$ has almost the same ASE; however, $\zeta =4$ leads to the ASE reduction in all other single-cell processing schemes. Because, there is a trade-off between the NMSE and the SE when increasing $\zeta$. While increasing $\zeta$ reduces the NMSE, the duration of the data transmission phase is also reduced which is shown as a coefficient $(1 - \frac{\zeta K}{\tau_c})$ in \eqref{SE_UatF} and \eqref{SE_Exact}.}

\begin{figure}
\centering
\includegraphics[width=0.6\textwidth]{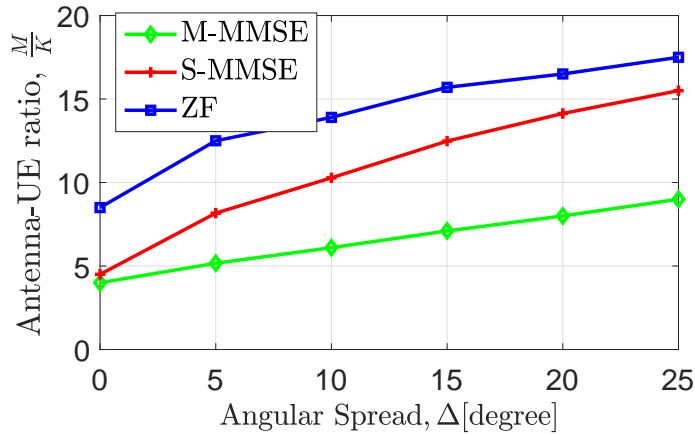}
\caption{Antenna-UE ratio $M/K$ needed for {SE} = 3 $[\rm bit/s/Hz/UE]$ as a function of angular spread $\Delta$ at $\lambda = 50$ $[\rm{BS/km^2}]$. The dual-slope path loss model is adopted. Also, $\zeta = 4$ 
 and $K=10$.}\label{fig:M/KvsASD}
\end{figure}

\begin{figure}
\centering
\includegraphics[width=0.6\textwidth]{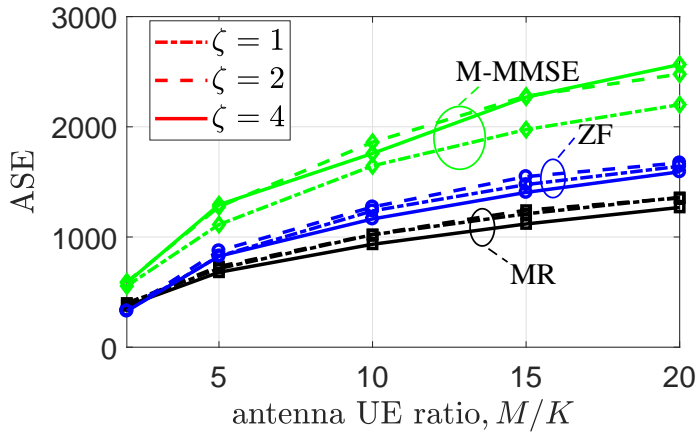}
 \caption{{ASE as a function of $M/K$  based on lower bound given by \eqref{SE_Exact} for given dual-slope path loss model in Table \ref{table1}  with  $\lambda = 50 [\rm{BS/km^2}]$, $\mathsf{SNR}_0 = 5$ dB and $\zeta = \{1, 2, 4\}$ for correlated Rayleigh fading with $\Delta = 10^\circ$.}}
\label{fig:ASEvsM/K}
\end{figure}


\section{Conclusion}\label{sec:conclusion}
We studied the impact of BS densification on the channel estimation accuracy and SE for the UL of the massive MIMO systems in the correlated Rayleigh fading.
First, theoretical analysis was done for MR and ZF in the uncorrelated fading which reveals the interplay between the network parameters. Particularly, we quantified the least $M/K$ based on $\lambda$ for which the pilot contamination dominates the (intra- and inter-cell) interference in the uncorrelated fading and showed that the pilot contamination plays a major role at any $\lambda$
for impractical values of $M/K$, i.e., $M/K > 10$.
In the general channels, we showed that the SE obtained by all combining schemes is a monotonic non-increasing function of the BS density, $\lambda$, while the M-MMSE provide notable SE in all $\lambda$ which is remarkable in the dense network, and in comparison with the uncorrelated Rayleigh fading, we observed that the M-MMSE performance degrades. Moreover, the resulting ASE is a non-decreasing function of $\lambda$.  We also concluded that as the channel becomes correlated, we need lower antenna-UE ratio to
achieve the desired performance in terms of SE.


\appendix \label{App} 


\begin{center}
\textbf{Appendix A. Proof of Corollary \ref{corollary_NMSE}}
\end{center}
{Regarding \eqref{correlation_error} and \eqref{channelEst_Uncorr}, we have $\tr\{\vect{C}_{jk}^j\} = M (\beta_{jk}^j - \gamma_{jk}^j) $.} Then, we can compute the expectation on pilot allocations by applying Jensen's inequality as follows:
\begin{equation}
\mathbb{E}_{{\{\bf a\}}} \left\{\gamma_{jk}^{j}\right\} \geq  \beta_{jk}^{j}  \left( 1 +  \frac{1}{\zeta} \sum\limits_{l \in \vect{\Psi}_\lambda \setminus \{j\} }  \frac{\beta_{l k}^{j}}{\beta_{l k}^{l}} +   \frac{1}{\tau_p \mathsf{SNR}_{\tr}} \right)^{-1},
\end{equation}
and 
 $\mathbb{E}_{{\{\bf d\}}} \left\{\sum\limits_{l \in \vect{\Psi}_\lambda \setminus \{j\} } \frac{\beta_{l k}^{j}}{\beta_{l k}^{l}} \right\}$ can be computed using \cite{fahim, pizzo2018network}. 

\begin{center}
\textbf{Appendix B. Proof of Lemma \ref{lemma_increasingNMSE}}
\end{center}
In order to prove that NMSE is  a non-decreasing function of $\lambda$, it is sufficient to show 
$\frac{\partial \mu_{\kappa}}{\partial \lambda} > 0, \kappa = 1, 2$ as follows\footnote{We prove it for both $\kappa = \{1, 2\}$, because we need to know the behavior of $\mu_{1}$ and $\mu_{2}$ for other analyses.}:
\begin{eqnarray}
\frac{\partial \mu_{\kappa}}{\partial \lambda}  = \sum\limits_{n=1}^N \eta_{1, n}(\kappa)  + \eta_{2, n}(\kappa)  +\eta_{3, n}(\kappa), \kappa = 1, 2,
\end{eqnarray}
with:
\begin{eqnarray}
\nonumber
\begin{cases}
\eta_{1, n}(\kappa) = 2\pi^2 \lambda \frac{-r_{n-1}^4 e^{-\pi \lambda r_{n-1}^2} + r_{n}^4 e^{-\pi \lambda r_{n}^2}}{\kappa \alpha_n - 2}, \\
\eta_{2, n}(\kappa)  = \frac{2 c_n(\kappa)(1 - \frac{\kappa \alpha_n}{2})}{(\pi \lambda)^{\frac{\kappa \alpha_n}{2}}} \left(  \Gamma\left(1+\frac{\kappa \alpha_n}{2}, \pi \lambda r_{n-1}^2\right) -   \Gamma\left(1+\frac{\kappa \alpha_n}{2}, \pi \lambda r_{n}^2\right)\right), \\
 \eta_{3, n}(\kappa)  =  2 \pi^2 \lambda c_n(\kappa)  ( -r_{n-1}^{2 + \kappa \alpha_n} e^{-\pi \lambda r_{n-1}^2} + r_{n}^{2 + \kappa \alpha_n} e^{-\pi \lambda r_{n}^2}),
\end{cases}
\end{eqnarray}
whereas we use Leibniz's rule in the computation of $\eta_{1, n}(\kappa) $ and $\eta_{3, n}(\kappa)$.
\subsubsection*{\textbf{Investigation of the behavior of $\eta_{1, n}(\kappa)$}}
We define $f(r; \lambda)  \triangleq r^4 e^{-\pi \lambda r^2}, r \geq 0$.
It is shown that $f(r; \lambda) \geq 0$ is a non-monotonic function.
We can easily prove this behavior through some computations on its derivative for each $\lambda$. We discover that:
\begin{align*}
\eta_{1, n}(\kappa) =  2\pi^2 \lambda \frac{-f(r_{n-1}; \lambda) + f(r_{n}; \lambda)}{\kappa \alpha_n - 2}, 
\end{align*}
is not a positive function for each $n$; however, we can prove that $\sum\limits_{n=1}^N \eta_{1, n}(\kappa) > 0$.
 If we pay attention to $\eta_{1, n}(\kappa)$ and  $\eta_{1, (n+1)}(\kappa)$ simultaneously, it is understood that $\frac{f(r_n; \lambda)}{\kappa \alpha_n - 2} > \frac{f(r_n; \lambda)}{\kappa \alpha_{n+1} - 2}$ due to that fact that, $\kappa \alpha_n - 2 < \kappa \alpha_{n+1} - 2$ and finally, we have $\sum\limits_{n=1}^N \eta_{1, n}(\kappa) > 0$. 

\subsubsection*{\textbf{Investigation of the behavior of $\eta_{2, n}(\kappa)$}}
Obviously, $c_n(\kappa) \leq 0$ and because $\Gamma(a, x_1) > \Gamma(a, x_2)$ for $x_1 < x_2$, we find out that 
$\sum\limits_{n=1}^N \eta_{2, n}(\kappa) > 0$.

\subsubsection*{\textbf{Investigation of the behavior of $\eta_{3, n}(\kappa)$}}
Defining $ g(r; \alpha_n, \kappa) = r^{2 + \kappa \alpha_n} e^{-\pi \lambda r^2}$ and comparing it with the function $f(r; \lambda)$, we can realize that $ g(r; \alpha_n, \kappa) < f(r; \lambda)$. Considering
$|c_n(\kappa)| < 1$, we conclude that $|\eta_{3, n}(\kappa)| < |\eta_{1, n}(\kappa)|$.
Finally, by considering above inequalities, we reach our aim.


\begin{center}
\textbf{Appendix C. Proof of Corollary \ref{corollary_optimalPRF}}
\end{center}
By taking the derivative of \eqref{R_inf} yields
\begin{eqnarray}
\frac{\partial R_{M\to \infty}}{\partial \zeta} = -\frac{K}{\tau_c} \log_2\left(1 + \frac{\zeta}{\mu_2}\right) + \frac{1 - \zeta K / \tau_c}{\mu_2 \ln(2) (1 + \zeta/ \mu_2)}.
\end{eqnarray}
Setting $ x = \frac{\tau_c/ K - \zeta }{\mu_2 + \zeta}$ we obtain
$ (x + 1) e^{x+1} = e (\frac{\tau_c}{K \mu_2} + 1)$, whose solution is found as $ x = W( e ( \frac{\tau_c}{K \mu_2} + 1)) - 1$.

\begin{center}
\textbf{Appendix D. Proof of Corollary \ref{corollary_ComparingReductionRate}}
\end{center}
To prove this, we compare the derivation of SINR w.r.t $\lambda$.
\begin{eqnarray}\nonumber
 \left|\frac{\partial {\mathsf{SINR}^{\rm UatF, MR}}}{\partial \lambda}\right|  &=&  \frac{a_1 \frac{\partial A(\mu_1)}{\partial \lambda} + a_2 \frac{\partial \mu_1}{\partial \lambda} + a_3 \frac{\partial \mu_2}{\partial \lambda}}{(1 / \mathsf{SINR}^{\rm UatF, MR})^2}, \\ \nonumber
 \left|\frac{\partial {\mathsf{SINR}^{\rm UatF, ZF}}}{\partial \lambda}\right| &=&  \frac{b_1 \frac{\partial A(\mu_1)}{\partial \lambda} + b_2 \frac{\partial \mu_1}{\partial \lambda} + b_3 \frac{\partial \mu_2}{\partial \lambda}}{(1 / \mathsf{SINR}^{\rm UatF, ZF})^2},
\end{eqnarray}
where the coefficients for MR and ZF are as follows:
\begin{align*}
\begin{cases}
a_1 = \frac{1}{M \mathsf{SNR}_0} + \frac{K}{M}(1 + \mu_1), 
a_2 = \frac{K}{M} A(\mu_1), 
a_3 = \frac{1}{\zeta} ( \frac{K}{M} + 1), \\
b_1 = \frac{1}{(M-K) \mathsf{SNR}_0} + \frac{K}{M-K}(1 + \mu_1), 
b_2 = \frac{K}{M-K} A(\mu_1), 
b_3 = \frac{1}{\zeta}.
\end{cases}
\end{align*}
We know $M \gg K$ and as we mentioned before in Lemma \ref{lemma_increasingNMSE}, $\frac{\partial A(\mu_1)}{\partial \lambda} \geq 0$, $\frac{\partial \mu_1}{\partial \lambda} \geq 0$ and $\frac{\partial \mu_2}{\partial \lambda} \geq 0$.
Then it would be concluded that $a_1 < b_1$, $a_2 < b_2$ and $a_3 \approx b_3$. On ther other hand,  $\mathsf{SINR}^{\rm UatF, MR} \leq \mathsf{SINR}^{\rm UatF, ZF}$. These lead to the inequality $ \frac{\partial {\mathsf{SINR}^{\rm UatF, MR}}}{\partial \lambda} <  \frac{\partial {\mathsf{SINR}^{\rm UatF, ZF}}}{\partial \lambda}$.


\bibliographystyle{IEEEtran}

\bibliography{IEEEabrv,ref,refbook,ref1}

\end{document}